\documentclass[10pt, conference, compsocconf]{IEEEtran}
\IEEEoverridecommandlockouts
\usepackage{ifthen}

\usepackage{algorithm}
\usepackage{algorithmicx}
\usepackage[noend]{algpseudocode}

\ifCLASSINFOpdf
\else
\fi
\usepackage[cmex10]{amsmath}
\usepackage{amssymb, amsthm, amsfonts}

\newtheorem{theorem}{Theorem}[section]
\newtheorem{lemma}{Lemma}[section]
\newtheorem{definition}{Definition}[section]

\usepackage[capitalize]{cleveref}

\newcommand{\eps}{\varepsilon}

\newcommand{\floor}[1]{\lfloor #1 \rfloor}
\newcommand{\ceil}[1]{\lceil #1 \rceil}

\newcommand{\N}{\mathbb{N}}

\newcommand{\F}{\mathbb{F}}
\newcommand{\hash}{\mathtt{hash}}
\newcommand{\hashRo}{\ensuremath{\mathtt{hash^o_{R}}}}
\newcommand{\C}{\mathcal{C}}

\newcommand{\exclude}[1]{}

\newcommand{\noSTOC}[1]{}
\newcommand{\includingdet}[1]{#1}

\newcommand{\FullOrShort}{full}

\ifthenelse{\equal{\FullOrShort}{full}}{
	
  \renewcommand{\baselinestretch}{1}
  \newcommand{\fullOnly}[1]{#1}
	
  \newcommand{\shortOnly}[1]{}
  \newcommand{\algorithmsize}{\normalsize}
  }{
    \renewcommand{\baselinestretch}{0.98}
    \newcommand{\fullOnly}[1]{}
		
		\newcommand{\shortOnly}[1]{#1}
    \newcommand{\IncludePictures}[1]{}
   \newcommand{\algorithmsize}{\footnotesize}
  }

\begin{document}
%
\title{Optimal Document Exchange\includingdet{\\ \large and New Codes for Insertions and Deletions}}



\author{\IEEEauthorblockN{Bernhard Haeupler\thanks{Supported in part by NSF grants CCF-1618280, CCF-1814603, CCF-1527110, NSF CAREER award CCF-1750808, and a Sloan Research Fellowship.}}
\IEEEauthorblockA{Computer Science Department\\
Carnegie Mellon University\\
Pittsburgh, USA\\
\texttt{haeupler@cs.cmu.edu}}
}


%


\maketitle

\begin{abstract}
We give the first communication-optimal document exchange protocol. For any $n$ and $k < n$ our randomized scheme takes any $n$-bit file $F$ and computes a $\Theta(k \log \frac{n}{k})$-bit summary from which one can reconstruct $F$, with high probability, given a related file $F'$ with edit distance $ED(F,F') \leq k$.


The size of our summary is information-theoretically order optimal for all values of $k$, giving a randomized solution to a longstanding open question of \cite[Orlitsky; FOCS'91]{orlitsky1991interactive}. It also is the first non-trivial solution for the interesting setting where a small constant fraction of symbols have been edited, producing an optimal summary of size $O(H(\delta)n)$ for $k=\delta n$. This concludes a long series of better-and-better protocols which produce larger summaries for sub-linear values of $k$ and sub-polynomial failure probabilities. In particular, the recent break-through of \cite[Belazzougui, Zhang; FOCS'16]{belazzougui2016edit} assumes that $k < n^\epsilon$, produces a summary of size $O(k\log^2 k + k\log n)$, and succeeds with probability $1-(k \log n)^{-O(1)}$. %
\includingdet{

We also give an efficient derandomized document exchange protocol with summary size $O(k \log^2 \frac{n}{k})$. This\footnote{The same derandomization result was simultaneously and independently discovered by \cite[Cheng, Jin, Li and Wu; FOCS'18]{docexchangefocs}. 
Both works were put on arxiv days apart~\cite{docexchangearxivhaeupler,docexchangearxivcheng}. However, the author served on the program committee of FOCS'18 and was as such not permitted to submit his work there.}~improves, for any $k$, over a deterministic document exchange protocol by Belazzougui~\cite{belazzougui2015efficient} with summary size $O(k^2 + k \log^2 n)$. Our deterministic document exchange directly provides new efficient systematic error correcting codes for insertions and deletions. These (binary) codes correct any $\delta$ fraction of adversarial insertions/deletions while having a rate of $1 - O(\delta \log^2 \frac{1}{\delta})$ and improve over the codes of Guruswami and Li and Haeupler, Shahrasbi and Vitercik which have rate $1 - \Theta\left(\sqrt{\delta} \log^{O(1)} \frac{1}{\epsilon}\right)$. 
}

\end{abstract}


%
\IEEEpeerreviewmaketitle

\section{Introduction}\label{sec:intro}

This paper gives the first efficient solution for the document exchange problem with an order optimal summary size. Our efficient randomized hashing scheme takes any $n$-bit file $F$ and for any $k$ computes an optimal sized $O(k \log \frac{n}{k})$-bit summary from which one can reconstruct $F$ given a related file $F'$ with edit distance\footnote{The edit distance $ED(F,F')$ between two strings $F$ and $F'$ is the minimal number of insertions, deletions, or symbol changes that transform one string into the other.} $ED(F,F') \leq k$. We also give a near optimal derandomization which deterministically computes an $O(k \log^2 \frac{n}{k})$ bit summary. This leads to improved systematic binary error correcting codes which efficiently correct any $\delta$ fraction of adversarial insertions and deletions while achieving a near optimal rate of $1 - O(\delta \log^2 \frac{1}{\delta})$.


Document exchange, or remote data synchronization, is an important problem in practice that frequently occurs when synchronizing files across computer systems or maintaining replicated data collections over a bandwidth limited network. In the simplest version it consists of two machines that each hold a copy of an $n$-bit file $F'$ where on one machine this file may have been updated to $F$. When updating the data on the other machine one would ideally like to only send information about their differences instead of sending the whole file $F$. This is particularly important because network bandwidth and data transfer times are limiting factors in most applications and $F$ often differs little from $F'$, e.g., only a small number $k$ of changes have been applied or a small fraction of the content has been edited, i.e., $k=\delta n$ for some small constant $\delta \in (0,1)$. This ``scenario arises in a number of applications, such as synchronization of user files between different machines, distributed file systems, remote backups, mirroring of large web and ftp sites, content distribution networks, or web access [over a slow network]''~\cite{irmak2005improved}.


One can imagine a multi-round protocol in which the two machines adaptively figure out which parts of the outdated file have not been changed and need not be transmitted. However, multi-round protocols are too costly and not possible in many settings. They incur long network round-trip times and if multiple machines need updating then a separate synchronization protocol needs to be run for each individual such machine. Surprisingly, Orlitsky~\cite{orlitsky1991interactive}, who initiated the theoretical study of this problem in 1991, proved that the party knowing $F$ can compute a short summary of $\Theta(k \log \frac{n}{k})$ bits which can then be used by \emph{any} other party $i$ knowing a file $F'_i$, which differs from $F$ by $k$ potentially very different edits to recover $F$ and the $k$ edits that have been applied to obtain $F$ from $F'_i$. This is initially quite surprising especially because the summary is, up to constants, of equal size as a description of the unknown changes themselves. Indeed, an exchange of $\Omega(k \log \frac{n}{k})$ bits is information-theoretically necessary to describe the difference between two $n$-bit strings of edit distance $k$. Unfortunately however, Orlitsky's result is merely existential and requires exponential time computations for recovery, which prevents the result to be of practical use. Pointing out several powerful potential applications, Orlitsky left the question of an efficient single-round document exchange protocol that matches the non-constructive $O(k \log \frac{n}{k})$ summary size as an open question which has since inspired many theoretical and practical results working towards this goal. This paper solely focuses on such single-round document exchange schemes. For simplicity, like all other prior works, we assume that a good upper bound $k$ on the edit distance $ED(F,F')$ is known\footnote{Alternatively starting with $k=1$ and doubling $k$ until the recovery was successful leads to the same amount of communication, up to a factor of two, since the summary size only depends linearly on $k$.}.



In practice rsynch~\cite{tridgell1996rsync} has become a widely used tool to achieve efficient two-round document exchange / file synchronization while minimizing the amount of information sent. Rsynch is also used as a routine in the rdiff tool to efficiently compute changes between files, e.g., in version control systems. Many similar protocols have been suggested and implemented. Unfortunately rsynch and almost all other tools do not have any guarantees on the size of the data communicated and it is easy to give examples on which these algorithms perform extremely poorly. A notable exception is a scheme of 
Irmak, Mihaylov and Suel~\cite{irmak2005improved}.


On the theoretical side the protocol of \cite{irmak2005improved} was the first computationally efficient single-round document exchange scheme with a provable guarantee on the size of a summary in terms of the edit distance $k = ED(F,F')$ achieving a size of $O(k \log \frac{n}{k} \log n)$. Independently developed fuzzy extractors~\cite{dodis2008fuzzy} can also be seen as providing a document exchange scheme for some $k$ polynomially small in $n$. A randomized scheme by Jowhari~\cite{jowhari2012efficient} independently achieved  a size of $O(k \log n \log^* n)$. In two recent break-throughs Chakraborty, Goldenberg, and Kouck\`y~\cite{chakraborty2016streaming} designed a low distortion embedding from edit distance to hamming distance which can be used to get a summary of size $\Theta(k^2 \log n)$ and Bellazougi and Zhang~\cite{belazzougui2016edit} further build on this randomized embedding and achieved a scheme with summary size $\Theta(k \log^2 k  + k \log n)$ which is order optimal for\footnote{We write $\exp(x)$ as a shortcut for $\Theta(1)^x = 2^{\Theta(x)}$.} $k = \exp(\sqrt{\log n})$. All of these schemes are randomized. The first deterministic scheme with summary size $\Theta(k^2 + k \log^2 n)$ was given by Belazzougui~\cite{belazzougui2015efficient}. All these document exchange schemes have some sub-linear restriction on the maximal magnitude of $k$. For example, the breakthrough result of \cite{belazzougui2016edit} assumes that $k < n^\eps$ for some sufficiently small constant $\eps >0$. In particular, there does not exist a scheme which works for the interesting case where the edit distance $k$ is a small constant fraction $\delta$ of the file length, e.g., if $1\%$ of the content has been edited. 


Deterministic document exchange solutions are furthermore related to error correcting codes for insertions and deletions, another  topic that has been studied quite intensely recently~\cite{schulman1999asymptotically,guruswami2015deletion,GL-isit16,HaeuplerSTOC17p46,haeupler2017synchronization2,haeupler2017synchronization3,brakensiek2016efficient,bukh2016improved} (see also these surveys~\cite{sloane2002single,mitzenmacher2009survey}). In particular, as we will see later, any single-round deterministic document exchange scheme with summary size $s = \Theta(|S_F|)$ for edit distance $\Theta(k)$ is equivalent\footnote{This equivalency does not hold for randomized document exchange schemes or for non-systematic error correcting codes for insertions and deletions.} to an systematic error correcting block code with redundancy $\Theta(s)$ which can correct up to $\Theta(k)$ errors. Through this equivalence one can derive a systematic insdel code with redundancy $O(k^2 + k \log^2 n)$ from the deterministic document exchange of Belazzougui~\cite{belazzougui2016edit}. A non-systematic code  with redundancy $O(k^2 \log k \log n)$ was given by by Brakensiek, Guruswami and Zbarsky~\cite{brakensiek2016efficient} but this code is only efficient for constant $k$. Most other works on error correcting insdel codes~\cite{schulman1999asymptotically,guruswami2015deletion,GL-isit16,HaeuplerSTOC17p46,haeupler2017synchronization2,haeupler2017synchronization3} have focused on the setting where a constant fraction of symbols have been corrupted, i.e., $k = \delta n$. For large finite alphabets Haeupler and Shahrasbi~\cite{HaeuplerSTOC17p46} gave insdel codes with optimal redundancy (up to a $(1+\eps)$ factor) and for binary codes Guruswami and Li~\cite{GL-isit16} and Haeupler, Shahrasbi and Vitercik~\cite{haeupler2017synchronization2} gave efficient codes for insertions and deletions with redundancy $\Theta\left(\sqrt{\delta} \log^{O(1)} \frac{1}{\delta} \cdot n\right)$ for any sufficiently small constant $\delta$.

\section{Our Results}

We positively answer the 28 year old open question of Orlitsky~\cite{orlitsky1991interactive} asking for an efficient document exchange scheme matching the optimal summary size of the existential results, at least for the randomized case, and give an efficient randomized hashing scheme for the single-round document exchange problem which, for any $k$, produces a summary of order optimal size $\Theta(k \log \frac{n}{k})$.

\begin{theorem}\label{thm:mainrandomized}
For any $k>0$ there is a randomized algorithm, which given any $n$-bit string $F$ produces a $\Theta(k \log \frac{n}{k})$-bit summary $S_F$. There also is a deterministic recovery algorithm, which given $S_F$ and any string $F'$ that is independent from the randomness used for computing $S_F$ and satisfies $ED(F,F')\leq k$, recovers the string $F$ with high probability, i.e., with probability $1 - n^{-O(1)}$. 
\end{theorem}

This improves over the recent break-through of \cite{belazzougui2016edit} which produces a summary of size $\Theta(k\log^2 k + k\log n)$ and works as long as $k < n^\eps$. We remark that the scheme in \cite{belazzougui2016edit} has a failure probability which is polylogarithmic in $n$ and polynomial in $k$ whereas the scheme in \Cref{thm:mainrandomized} works with high probability. If one wants to boost the scheme in \cite{belazzougui2016edit} to work with high probability one needs to send $\frac{\log n}{\log k + \log \log n}$ independent summaries making the overall summary size (up to a $\log \log n$ for sub-logarithmic $k$) equal to $\Theta(k \log k \log n + k \frac{\log^2 n}{\log k})=\Omega(k \log^{1.5} n)$.

As a precursor to our main result we obtain a document exchange protocol with sub-optimal summary size that has the advantage that it can be efficiently derandomized\footnote{The same derandomization result was simultaneously and independently discovered by \cite[Cheng, Jin, Li and Wu; FOCS'18]{docexchangefocs}. 
Both works were put on arxiv days apart~\cite{docexchangearxivhaeupler,docexchangearxivcheng}. However, the author served on the program committee of FOCS'18 and was as such not permitted to submit his work there.}. This gives a deterministic document exchange protocol with summary size $\Theta(k \log^2 \frac{n}{k})$, improving over the deterministic scheme by Belazzougui\cite{belazzougui2015efficient} with summary size $\Theta(k^2 + k \log n)$.


\begin{theorem}\label{thm:maindet}
There is a deterministic document exchange algorithm, which given any any $k>0$ and any $n$-bit string $F$ produces a $\Theta(k \log^2 \frac{n}{k})$-bit summary $S_F$, such that a deterministic recovery algorithm, which is given $S_F$ and any string $F'$ with edit distance $ED(F,F') \leq k$, recovers $F$. 
\end{theorem}

The schemes from \Cref{thm:mainrandomized} and \Cref{thm:maindet} are the first document exchange protocols which work for the interesting setting in which a constant fraction of edits need to be communicated. In particular, if the edit distance between $F$ and $F'$ is $\delta n$ for some small constant $\delta > 0$ then our optimal randomized scheme produces a summary of size $\Theta(n \delta \log 1/\delta) = \Theta(n H(\delta)) \ll n$, where $H(.)$ is the binary entropy function. Our deterministic scheme incurs another $\log \frac{1}{\delta}$ factor but the summary size of $\Theta(n \delta \log^2 \frac{1}{\delta})$ is still much smaller than $n$ for sufficiently small $\delta$. 


As mentioned above, efficient deterministic document exchange protocols are known to be equivalent to efficient, systematic error correcting codes for insertions and deletions. Via this equivalency the deterministic document exchange scheme from \Cref{thm:maindet} directly gives the following near optimal, efficient, systematic error correcting codes for insertions and deletions which work for any $k$ adversarial insertions and deletions with near optimal redundancy $\Theta(k \frac{\log^2 \frac{n}{k}}{\log q} + k)$. 

\begin{theorem}\label{thm:maincodes}
For any $n$, any $k < n$, and any finite alphabet $\Sigma$ of size $q = |\Sigma| \geq 2$ there is a simple deterministic encoding algorithm which takes an input string $X \in \Sigma^n$ and outputs a systematic codeword $C(X) \in \Sigma^{n+r}$ consisting of $X$ itself and $r=\Theta(k \frac{\log^2 \frac{n}{k}}{\log q} + k)$ redundant extra symbols and a deterministic decoding algorithm $A_{dec}$ such that for any $X$ and any $C'$ with $ED(C',C(X)) \leq k$ we have $A_{dec}(C')=X$, i.e., one can recover $X$ from any corrupted codeword $C'$ which is $k$-close to $C(X)$ in edit distance. 
\end{theorem}

This is an almost quadratic improvement in terms of redundancy and rate loss compared to the state-of-the-art binary insdel codes of Brakensiek, Guruswami and Zbarsky~\cite{brakensiek2016efficient} and Belazzougui~\cite{belazzougui2016edit} for small values of $k$ and the codes of Guruswami et al.~\cite{guruswami2015deletion,GL-isit16} and Haeupler, Shahrasbi and Vitercik~\cite{haeupler2017synchronization2,HaeuplerSTOC17p46} for the case of a constant fraction of corruptions. A more detailed comparison is given in \Cref{sec:derandomization}.

\section{Hash Functions and Summary Structure}

In this section we describe and define the simple inner-product hash functions used in our schemes and the content and structure of the summary $S_F$ of $F$. We start by giving some intuition about the summary structure in \Cref{sec:summaryinformal}, give our string notation in \Cref{sec:stringnotation}, formally define our hash function in \Cref{sec:defhashfunctions}, and define our summary structure in \Cref{sec:summary}.

\subsection{Intuition for the Summary Structure}\label{sec:summaryinformal}

Essentially all document exchange algorithms used in practice, including rsynch, use the very natural idea of cutting the file $F$ into blocks and sending across hashes of these blocks in order to identify which of these blocks are contained in $F'$ without any edits. 

Once identical blocks have been identified the remaining information containing all differences between $F$ and $F'$ is small and can be transmitted. The protocol of Irmak, Mihaylov and Suel~\cite{irmak2005improved} also follows this strategy. However, it does not use a fixed block length but uses $\log  \frac{n}{k}$ levels\footnote{All logarithms in this paper are with respect to basis two unless stated otherwise.} of exponentially decreasing block length. This allows to progressively zoom into larger blocks containing some edits to identify smaller blocks within them that do not contain an edit. Given that higher levels should already identify large parts of the string $F$ that are identical to $F'$ and thus known to the party reconstructing $F$ many of the hashes of lower levels will not be of interest to the party reconstructing $F'$. To avoid having to send these hashes one could run an adaptive multi-level protocol in which the reconstructing party provides feedback at each level. A great and much simpler single-round alternative introduced by~\cite{irmak2005improved} is to use (the non-systematic part of) systematic error correcting codes which allows the receiving party to efficiently reconstruct the hashes it is missing, without the sending party needing to know which hashes these are. The summary $S_F$ now simply consists of these encodings of hashes of all $\log \frac{n}{k}$ levels and this summary can be sent to the reconstructing party in a single-round document exchange protocol. Given that at most $ED(F,F')$ blocks can be corrupted in each level the summary size is $\Theta(k \log \frac{n}{k} \cdot o)$ where $o$ is the size of a single hash. Using randomized $o=\Theta(\log n)$-bit hashes with $\Theta(\log n)$-bit seeds, which are guaranteed to be correct with high probability, leads to the $\Theta(k \log \frac{n}{k} \cdot \log n )$ bit summary size of \cite{irmak2005improved}.
 

The hashing schemes in this paper mostly follow the same practical framework. In fact, the summary structure we use for our simpler (sub-optimal deterministic) document exchange protocol with summary size $\Theta(k \log^2 \frac{n}{k})$ is identical to \cite{irmak2005improved} except that we use a smaller hash size $o = \Theta(\log \frac{n}{k})$ and a  compact way to describe the randomness used for hashing, which can then also be used for derandomization. Our main result further reduces the hash size $o$ to merely a fixed constant. This requires an 
much more robust recovery algorithm which can deal with high hash collision probabilities. In particular, since the failure probability of a hash is exponential in the hash size $o$, choosing $o=\log n$ as in \cite{irmak2005improved} implies that no hash collision happens with high probability and choosing $o=\log \frac{n}{k}$ still keep the expected number of hash collisions at $O(k)$, that is, at the same order of magnitude as the errors one has to deal with anyway. The fact that recovery from constant size hashes with a constant failure probability is even existentially possible requires a much more intricate probabilistic analysis. Furthermore, the key trick used in \cite{irmak2005improved,belazzougui2015efficient,docexchangearxivcheng} to use error correcting codes to obliviously communicate the missing or incorrect hashes in each round inherently requires $\Theta(k \log \frac{n}{k})$-bits to be sent in each of the $\log \frac{n}{k}$ levels. This forms another serious barrier that needs to be overcome for our main result. 

\subsection{String Notation}\label{sec:stringnotation}

Next we briefly give the string notation we use throughout.

Let $S \in \Sigma^*$ be a string. We denote with $|S|$ the length of $S$ and for any $i,j \in [1,|S|]$ with $i \leq j$ we denote with $S[i,j]$ the substring of $S$ between the $i^{th}$ and $j^{th}$ symbol, both included. A sub-string of $S$ is always a set of consecutive symbols in $S$, i.e., any string of the form $S[i,j]$. We also use multi-dimensional arrays of symbols, in which every index typically begins with $0$. An array position $H[i,j]$ with $i,j \in \mathbf{N}_0$ can either contain a symbol over some alphabet $\Sigma'$ or be empty. We denote with $H[i,.]$ the string of symbols $(H[i,1], H[i,2], \ldots)$ containing all symbols of the form $H[i,j]$.

\subsection{Inner-Product Hash Function $\hashRo$ and Randomness Table $R$}\label{sec:defhashfunctions}

Next we describe our hash function $\hash$, which computes nothing more than some $\F_2$-inner-products between its input string and some random bits. 

To properly keep track of the randomness used we will think of the randomness being supplied by a three dimensional table of bits we call $R$. We remark that our algorithms do not actually need to instantiate or compute $R$ explicitly. Instead the description of the bits contained in $R$ will be so simple that they can be generated/computed on the spot whenever needed. 

In addition to the string $S \in \Sigma^*$ to be hashed we supply four more arguments to $\hash$. A call like $\hashRo(S,s,\ell)$ produces a hash of the string $S \in \Sigma^*$ using the randomness table $R$. The parameter $o \in \N$ denotes the size of the hash, i.e., the number of bits that are produced as an output. The parameters $s$ and $\ell$ denote to which starting position and level the hash belongs to, respectively. These parameters are used to describe where in the randomness table $R$ to pull the randomness for the inner product from. This ensures firstly that hashes for different levels and intervals use different or ``fresh'' random bits and secondly that summary creation and recovery consistently use the same parts of $R$ when testing whether two strings stem from the same interval in the original string $F$. The inner product computed by $\hashRo(S,s,\ell)$ is now simply the $o$-bit string $h_1,\ldots,h_o$ for which $h_i = \bigoplus_{j=1}^{|S|} \left(S[j] \cdot R[s+j-1,\ell,i]\right)$.

Note that if $R$ is filled with independent and uniformly distributed bits we have that any two non-identical strings have colliding hashes with probability $2^{-o}$, i.e., for every $k \in N$, $S \neq S' \in \Sigma^k$ and $o,s,\ell \in \N$ it holds that
$$P_R[\hashRo(S,s,\ell) = \hashRo(S',s,\ell)] = 2^{-o}.$$
 The reason for this is that each of the $o$ output bits independently is an inner product between the string to be hashed and the same uniformly random bit string of length $|S|$ taken from $R$. Therefore the difference between $h_i$ and $h_i'$ is the inner product of a uniformly random string and a non-zero string and as such a uniformly distributed bit. The probability that each of the $o$ output bits is zero is now exactly $2^{-o}$.

Lastly, we add one further simplification to our hash function which is that if the output length $o$ is larger than the length $|S|$ of the string $S$ to be hashed then $\hashRo(S,s,\ell)$ simply outputs $S$ as a ``hash'', possibly padded with zeros. This gives a collision probability of zero for any two same-length strings with length at most $o$ and allows to read off the string $S$ from its hash.


\subsection{Summary Structure and Construction}\label{sec:summary}

In this section we formally describe and define the summary structure and construction which follows the informal description given in Section~\ref{sec:summaryinformal}:

The summary algorithm takes the string $F$ it wants to summarize, the parameter $k$ which essentially governs how many hashes are provided per level (in coded form) and the parameter $o$ which determines the hash size of the hashes used. For simplicity of the description we assume that the length $n = |F|$ of the string $F$ equals $4 k \cdot 2^L$ for some integer $L$, i.e., $n$ is a multiple of $4k$ and $\frac{n}{4k}$ is a power of two. This assumption is without loss of generality: One can send the length of $F$ along with the summary and, for the hash computations, extend $F$ to a string of length $4 k \cdot 2^m$, with $L = \ceil{\log_2 \frac{n}{4k}}$, by adding zeros to the end. The recovery algorithm simply adds the same number of zeros to $F'$ during the recovery and removes them again in the end. 

The summary $S_F$ contains the following coded hashes organized into $L+1$ levels:
\vspace{-0.12cm}
\begin{itemize}
	\item  The level zero simply cuts $F$ into $4k$ equal size pieces and records the hashes for each piece. I.e., let $\forall i \in [0,4k-1]: \ H[0,i] = \hashRo(F[1 + i\cdot 2^L, (i+1)\cdot 2^L],i\cdot 2^L,0)$ and we include $H[0,.]$, which consists of $ko$ bits, in the summary. 
\smallskip
	\item For level $\ell \in [1,L]$ we cut $F$ into $4k \cdot 2^{\ell}$ equal size pieces, compute the hash for each piece to form $H[\ell,.]$. The hashes themselves however are too large to be sent completely. Instead our warm-up (deterministic) scheme encodes these hashes using an error correcting code $\C_j$ which is simply the non-systematic part of a systematic linear $[4k 2^{\ell}+100k,4k 2^{\ell},11k]$ error correcting code over $\F_{2^o}$. Such a code exists and many explicit constructions based on algebraic geometry are known if $o = \Theta(\log n/k)$. Our optimal document exchange protocol with $o = \Theta(1)$ requires a more sophisticated scheme which we describe in Section~\ref{sec:alg2}. The information included in $S_F$ for this scheme consists of one or multiple hashes of (subsets of) $H[\ell,.]$ of total size $\Theta(o' k)$ for some sufficiently large constant $o' > o$.
\end{itemize}

We remark that the hashes $H[\ell,.]$ of levels $\ell \geq L-\floor{\log_2 o}$ are hashes of strings of length at most $o$. In this case the hash function $\hashRo$ simply outputs the strings itself making $H[L,.]$ (or already $H[L-\floor{\log_2 o},.]$) essentially equal to $F$ itself. 

In addition to these coded hashes the summary $S_F$ also contains the length $|F|$ and a compact description of the randomness table $R$. Throughout this paper we will use the $\eps$-biased probability spaces of Naor and Naor~\cite{naor1993small} for this compact description. In particular, we prove for all our schemes that the bits in $R$ do not need to be independent uniform bits but that it suffices if they are sampled from a distribution with reasonably small bias. The often exploited fact that a sample point from an $\eps$-biased distribution over $n^{O(1)}$ bits can be described by only $O(\log \frac{1}{\eps} + \log n)$ bits allows us to give very compact descriptions of the randomness used and send these along in the summary $S_F$. In particular, we do \emph{not} need to assume that the summary construction and the summary recovery algorithm have any shared source of randomness.

\section{Recovery Algorithms}

This section describes our recovery algorithms. We start in Section~\ref{sec:matchings} by defining hash induced substring matchings, which form the basis for our algorithms and their analysis. We then describe our recovery algorithms. In Section~\ref{sec:alg1} we first describe a randomized algorithm which produces a summary of size $\Theta(k \log^2 \frac{n}{k})$. This is a good warm-up for our main result. It demonstrates the overall algorithmic structure common to both our recovery algorithms, introduces the basic probabilistic analysis used to analyze them, and makes it easier to understand the problems that need to be addressed when pushing both the algorithmic ideas and the analysis to the limit for our main result. We also show in \Cref{sec:derandomization} how to derandomize this scheme to obtain \Cref{thm:maindet}. 
Lastly, Section~\ref{sec:alg2} contains the order optimal randomized hashing scheme which uses constant size hashes and thus achieves the optimal summary size of $\Theta(k \log \frac{n}{k})$.

\subsection{Hash Induced Substring Matchings}\label{sec:matchings}

For every $n,k,\ell$ we say that two index sequences $i_1,\ldots,i_{k'},i'_1,\ldots,i'_{k'} \in [1,n-2^{L-\ell}+1]$ of length $k'$ are a level-$\ell$ size-$k'$ (sub-string) matching between two strings $F, F' \in \{0,1\}^n$ induced by $\hashRo$ if
\vspace{-0.12cm}
\begin{itemize}
	\item $i_1 \leq i_2 \leq \ldots \leq i_{k'}$,
	\item all $i$-indices are starting points of blocks that got hashed in $S_F$ in level $\ell$, i.e, $i_j-1$ is a multiple of $2^{L-\ell}$ for every $j \in [1,k']$, and
	\item hashes of the strings in blocks that are matched are identical (we also say the hashes are matching or consistent), i.e., for all $j \in [1,k']$ we have that $\hashRo(F[i_j,i_j+2^{L-\ell}-1],i_j-1,\ell)$ equals $\hashRo(F'[i'_j,i'_j+2^{L-\ell}-1],i_j-1,\ell)$. In the rare instances where we (temporarily) relax this requirement we speak of a \emph{non-proper} matching. 
\end{itemize}
Furthermore, we call such a matching
\vspace{-0.12cm}
\begin{itemize}
	\item \emph{monotone} if $i'_1 \leq i'_2 \leq \ldots \leq i'_{k'}$,
	\item \emph{disjoint} if intervals that are matched in $F'$ are not overlapping, i.e., we have for every $j,j' \in [1,k']$ with $j \neq j'$ that $|i'_{j} - i'_{j'}| \geq 2^{L-\ell}$.
	\item \emph{bad} or \emph{$k'$-bad} if for every matched blocks in $F$ and $F'$ the actual strings are non-identical (despite having identical hashes), i.e., if for all $j \in [1,k']$ we have that $F[i_j,i_j+2^{L-\ell}-1] \neq F'[i'_j,i'_j+ 2^{L-\ell}-1]$.
	\item \emph{$k$-plausible} if it is monotone and $|i_1 - i'_1| + |(|F| - i_{k'}) - (|F'| - i'_{k'})| + \sum_{j=1}^{k'-1} |(i_j - i'_j) - (i_{j+1} - i'_{j+1})| \leq k$. Note that a monotone matching is $k$-plausible if it can be explained by at most $k$ insertion and deletion operations. 	
\end{itemize}
We generally assume all our matchings to be proper, monotone, and disjoint and often omit these qualifiers. Whenever we talk about non-necessarily monotone, not-necessarily proper or not-necessarily non-disjoint matchings we explicitly label these matchings as \emph{non-disjoint}, \emph{non-proper} and/or \emph{non-monotone}. Furthermore, if the context allows it we sometimes omit mentioning the strings $F,F'$, the hash function $\hashRo$, or the level $\ell$ with respect to which a matching satisfies the above conditions. 

For any two strings $F,F'$ one can compute a monotone matching of maximum size in time linear in the length of the strings $n$ and polynomial in $k$ using a standard dynamic program. The same is true for a maximum size disjoint, bad, or $k$-plausible monotone matching.
It is furthermore likely that using the same techniques which transforms the standard $O(nk)$ dynamic program for edit distance into an $O(n + k^2)$ dynamic program~\cite{ukkonen1985algorithms} can also be used to obtain $O(n + k^{O(1)})$ algorithms for the above maximum monotone matching variants as well. 

%

\subsection{Algorithm 1: Simple Level-wise Recovery}\label{sec:alg1}

\begin{algorithm}[htb!]
\caption{Simple Recovery with $o = \Theta(\log \frac{n}{k})$ and Summary Size $O(k \log^2 \frac{n}{k})$}
\begin{algorithmic}[1]
\algorithmsize

\State get $H_F[0,.]$ from $S_F$
\Statex
\For {$\ell = 0$ to $L-1$}
	
	\Statex
	\State $M_l \gets$ largest level-$\ell$ disjoint monotone matching from hashes in $H_F[\ell,.]$ into $F'$ 
	
	\Statex
	\Comment{{\bfseries Recover level $\ell+1$ hashes}}
	\State $\tilde{H}_F[\ell+1,.] \gets$ guesses for level $\ell+1$ hashes using $M_{\ell}$ and $F'$
	\State $H_F[\ell+1,.] \gets \text{Decode}_{C_{\ell+1}}(\tilde{H}_F[\ell+1,.],\text{encoding of $H_F[\ell+1,.]$ from $S_F$})$

\EndFor	

	\Statex
	\State $F \gets H_F[L,.]$

\end{algorithmic}
\label{alg:SimpleRecovery}
\end{algorithm}

Our first recovery algorithm, which we call Simple Level-wise Recovery, is now easily given (see also the pseudo-code description of this algorithm, which is given as Algorithm~1): 

Assume that the recovery algorithm has recovered all $4k \cdot 2^\ell$ hashes $H_F[\ell,.]$ of $F$ at level $\ell$ correctly. Initially $\ell=0$ and this assumption is trivially true because these hashes are included in the summary $S_F$. Equipped with these $4k \cdot 2^\ell$ hashes the algorithm finds the largest monotone disjoint matching between the level $\ell$ blocks in $F$ and blocks in $F'$ of the same length. The recovery algorithm now \emph{guesses} the level $\ell+1$ hashes $H_F[\ell+1,.]$ of $F$ using $M_{\ell}$ and $F'$ as follows:  Each block of $F$ in level $\ell$ splits into exactly two blocks in level $\ell+1$. For any block in $F$ that is matched to a sub-string in $F'$ with an identical hash the recovery algorithm guesses that the strings in these blocks are also identical and computes the hashes for the two sub-blocks in level $\ell+1$ by applying $\hashRo$ to the appropriate sub-string in $F'$. If a block in $F$ is not matched one can fill in something arbitrarily as a guess or mark it as an erasure. The hope is that the vector of hashes $\tilde{H}_F[\ell+1,.]$ for level $\ell+1$ guessed in this way is close in Hamming distance to the correct hashes $H_F[\ell+1,.]$. If this is the case, concatenating $\tilde{H}_F(\ell+1,.)$ with the redundancy $\text{Enc}_{\ell+1}$ for level $\ell+1$ from $S_F$ and decoding this to the closest codeword in $C_\ell$ correctly recovers the level $\ell+1$ hashes $H_F[\ell+1,.]$ and allows the algorithm go proceed to the next iteration and level. In this way the recovery algorithm iteratively recovers the hashes for every level one by one until level $L$. In level $L$ blocks are of constant size and $\hashRo$ becomes the (padded) identity function such that one can read off $F$ from $H_F[L,.]$.

\subsection{Correctness of Algorithm 1 and $k$-Bad Matchings}

In this subsection we give a sufficient condition for the correctness of Algorithm $1$. In particular, we prove that if there is no $k$-bad self-matching in $F$, i.e., a size-$k$ bad monotone disjoint matching, between $F$ and itself under $\hashRo$, then Algorithm $1$ recovers $F$ correctly. In Appendix~\ref{sec:derandomization} we then show that hashes of size $o = \Omega(\log \frac{n}{k})$ and only little randomness in $R$ are sufficient to make the existence of such a witness unlikely, in fact, so little randomness that one can easily derandomize the algorithm.




To prove the correctness of Algorithm 1 we will argue that the matching computed in each level is sufficiently large and, in the absence of a $k$-bad self-matching, of sufficient quality to allow the recovery of the hashes for the next level using the redundancy in $S_F$. This allows the recovery algorithm then to proceed similarly with the next level.

It is easy to see that the matching computed is always large assuming that $F$ and $F'$ are not too different:

\begin{lemma}\label{lem:matchingsize}
Assuming that the hashes for level $\ell$ were correctly recovered, Algorithm 1 computes a matching of size at least $4k \cdot 2^l - ED(F,F')$ in level $\ell$.
\end{lemma}
\begin{proof}
Since $F$ and $F'$ differ by only $ED(F,F')$ insertion, deletions, or symbol corruptions and since each such edit can affect at most one block we know that we can look at the monotone matching which matches all blocks in $F$ which did not suffer from such an edit to its identical sub-string in $F'$. Since the hashes were correctly recovered and the hashes use the same parts of $R$ to compute the inner-producet hashes this is a valid monotone matching of size $4k \cdot 2^l - ED(F,F')$. Since Algorithm 1 computes the largest valid matching it finds a matching of at least this size. 
\end{proof}

\medskip

We would like to say that if in the summary $S_F$ random enough hash functions with a small enough collision probability are used, which usually result from a sufficiently unbiased $R$ and a large enough hash output length $o$, then most of the matching pairs computed by Algorithm 1 are correct, i.e., correspond to sub-strings of $F$ and $F'$ that are identical under $\hashRo$ with the randomness $R$ used. For any matching which contains too many pairs of substrings which are not-identical but have the same hashes we abstract out a witness which explains why $\hashRo$ failed. For this we focus on bad matching pairs that go between non-identical intervals in $F$ and $F'$ which do not contain any edits. This however is exactly a $k$-bad selfmatching in $F$. The advantage of looking at such a witness is that its existence only depends on $S_F$ (or $\hashRo$, $R$ and $F$) but not on $F'$.


%

\begin{lemma}\label{lem:nobadmatchingimpliesgoodmatches}
Assume that the hashes for level $i$ were correctly recovered and that for level $\ell$ there is no $k$-bad matching of $F$ to itself under $\hashRo$. Then the monotone matching computed by Algorithm 1 for this level matches at most $ED(F,F') + k$ non-identical blocks in $F$ and $F'$.
\end{lemma}
\begin{proof}
$F$ and $F'$ differ by only $ED(F,F')$ insertion, deletions, or symbol corruptions and  each such edit can affect at most one of the blocks in $F'$ that are matched to a non-identical block in $F$. Therefore there are at most $ED(F,F')$ such matches in the monotone matching computed by Algorithm 1. Furthermore, if we restrict ourselves to the matches between non-identical sub-strings in $F$ and $F'$
computed by Algorithm 1 which are not of this type it is true that each of these matches comes from matching a sub-string in $F$ to a sub-string in $F'$ which, due to having no edits in it, is identical do a sub-string in $F$. Since, by assumption, Algorithm 1 used the correctly recovered level $\ell$ hashes and computes a monotone disjoint matching these matches form a bad matching in $F$. By assumption this matching can be of size at most $k$ giving the desired bound of at most $ED(F,F') + k$ non-identical blocks matched in $F$ and $F'$ by Algorithm 1
\end{proof}
\medskip

Lastly, because we use error correcting codes with sufficiently large distance we can easily correct for the $k$ missing hashes and $2k$ incorrect hashes. 

\begin{lemma}\label{lem:algonecorrectnesswithoutkmatch}
Assume that for all levels there is no $k$-bad matching of $F$ to itself under $\hashRo$ used to compute $S_F$, which is given to Algorithm 1 as an input. Furthermore assume that the input file $F'$ satisfies $ED(F,F') \leq k$. Then, Algorithm 1 correctly outputs $F$. 
\end{lemma}
\begin{proof}
We will first show by induction on the level $\ell$ that Algorithm 1 correctly recovers the level $\ell$ hashes $H_F[\ell,.]$ of $F$ that were computed for $S_F$. For level $\ell=0$ this is trivial because these hashes are a part of $S_F$ and therefore given to Algorithm 1 as an input. For the induction step we want to show that the hashes for level $\ell+1$ will be correctly recovered assuming that this has successfully happened for the hashes for level $\ell$. Here Lemma~\ref{lem:matchingsize} guarantees that a matching of size $4k \cdot 2^l - ED(F,F')$ is computed which results in at most $k = ED(F,F')$ blocks in the level $\ell$ of $F$ having no match and therefore at most $2k$ hashes in level $\ell+1$ are assigned the erasure symbol ``?'' in the guessing step of Algorithm 1. Furthermore, the assumptions for Lemma~\ref{lem:nobadmatchingimpliesgoodmatches} are satisfied guaranteeing that at most $ED(F,F')+k = 2k$ of the matchings computed by Algorithm 1 in level $\ell$ belong to non-identical strings in $F$ and $F'$. This results in at most $4k$ of the hash values computed in the guessing step of Algorithm 1 being incorrect. The Hamming distance between the correct hashes $H_F[i+1,.]$ and the estimate $\tilde{H}_F[i+1,.]$ produced by Algorithm 1 is therefore at most $6k$. Given that the error correcting code used has distance $13k$ these errors will be corrected leading to a correct recovery of the level $\ell+1$ hashes in Step $5$ of Algorithm 1. In its last iteration Algorithm 1 will correctly recover the level $L$ hashes $H_F[L,.]$, which, as discussed at the end of Section~\ref{sec:summary}, is equal to $F$ (up to padding extra zeros to each hash). 
\end{proof}
\medskip

It remains to show that $k$-bad matchings are highly unlikely. 
\Cref{lem:okbiasednomatching} does exactly this. It shows that for sufficiently random $R$ and $o = \Omega(\log \frac{n}{k})$ with high probability no $k$-bad matching exists. We start by showing that a distribution whose bias is exponentially small in $k \log \frac{n}{k}$ suffices to avoid a $k$-bad matching with high probability. This is sufficient to guarantee the correctness of Algorithm 1. We furthermore shows how to extend these arguments to much lower quality distributions with a polynomially large bias. Since this second part is important for our derandomization but not needed to understand our main result we defer this second part to \Cref{sec:derandomization}.

\begin{lemma}\label{lem:okbiasednomatching}
For every sufficiently large $c \geq 1$ it holds that if $o = c \log \frac{n}{k}$ and $R$ is sampled from an $(2^{-2ok})$-biased distribution of bits then for every level $\ell$ the probability that there exists a $k$-bad self-matching of $F$ under $\hashRo$ is at most $2^{-\Omega(ok)}$.
\end{lemma}
\begin{proof}
Suppose for sake of simplicity that $R$ is sampled from iid uniformly random bits. In this case the probability for any individual sub-strings of $F$ to have a bad hash is $2^{-o}$. Furthermore, for a fixed $k$-bad matching the probability that all matching pairs are bad under $\hash$ is $2^{-ok}$. There furthermore exist at most $\binom{n}{k}^2 = 2^{O(k \log \frac{n}{k})}$ ways to choose the indices for a potential $k$-bad matching and therefore also at most $2^{O(k \log \frac{n}{k})}$ many such matchings. Taking a union bound over all these potentially bad matchings and choosing the constant $c$ large enough this guarantees that the probability that there exists a $k$-bad matching is at most $2^{-ok} \cdot 2^{O(k \log \frac{n}{k})} = 2^{-\Omega(ok)}$. 

Next we argue that the same argument holds if $R$ is $2^{-cok}$-biased. In particular, for every $k$- matching determining whether it is a $k$-bad matching only depends on the outcome of $ok$ linear tests on bits from $R$. For each of the $2^{ok}$ different outcomes for these tests the probability deviates at most by $2^{-2ok}$ from the setting where $R$ is sampled from iid uniformly random bits and is therefore still at most $2^{-ok}+2^{-2ok} = 2^{-\Omega(ok)}$. The same union bound thus applies. 
\end{proof}
\medskip

\subsection{Algorithm 2: Using Constant Size Hashes}\label{sec:alg2}

In this section we give a more sophisticated and even more robust recovery algorithm. Surprisingly this algorithm works even if the hash size $o$ used in the summary computations is a merely a small constant, leading to our main result, the order optimal randomized document exchange hashing scheme. 

The main difference between Algorithm~\ref{alg:SimpleRecovery} and Algorithm~\ref{alg:ImprovedRecovery} is that we take matchings of previous levels into account and restricting ourself to \emph{$k$-plausible} monotone matchings. This is sufficient to reduce the problem to a Hamming type problem on how to communicate the next level of hashes when most of them are already known to the receiver. However, here we cannot use systematic error correcting codes anymore but need to develop more efficient techniques. 


\begin{algorithm}[htb!]
\caption{Randomized Recovery with $o = O(1)$ and Summary Size $O(k \log \frac{n}{k})$}
\begin{algorithmic}[1]
\algorithmsize

\State $M_0 \gets \emptyset$ 
\State get $H_F[0,.]$ from $S_F$
\Statex
\For {$\ell = 0$ to $L-1$}
	
	\Statex
	\Comment{{\bfseries Transform into a proper level $\ell$ matching}}
	\State remove all matches not consistent with $H_F[\ell,.]$ from $M_{\ell}$

	\Statex
	\Comment{{\bfseries Compute plausible disjoint monotone matching for unmatched hashes}}
	\State $\Delta_l \gets$ largest $k$-plausible level-$\ell$ matching of hashes unmatched in $M_\ell$ into $F'$ 
	\State $M_{\ell+1} \gets M_l + \Delta_\ell$   (and split all edges into two making it a $(\ell+1)$-level matching)

	\Statex
	\Statex
	\Comment{{\bfseries Recover level $\ell+1$ hashes}}
	\State Recover $H_F[\ell+1,.]$ by using the substrings matched in $M_{\ell+1}$, guessing $\Theta(k)$ incorrect hashes and their values, and verifying correct guesses with hashes in $S_F$

\EndFor	

	\Statex
	\State $\tilde{F} \gets H_F[L,.]$ 

\end{algorithmic}
\label{alg:ImprovedRecovery}
\end{algorithm}

We note that the matchings $M_\ell$ produced by Algorithm~\ref{alg:ImprovedRecovery} are not necessarily monotone and not necessarily disjoint, i.e., they can contain matches to overlapping intervals in $F'$. At the beginning or end of an iteration the matching might even be a non-proper level $\ell$ matching in that there can be matches which stem from matching level $(\ell-1)$ hashes but do not have matching level $\ell$ hashes. At the beginning of an iteration such matches are removed making the matching proper. In order to analyze the progress of Algorithm~\ref{alg:ImprovedRecovery} we introduce the following notion of an okay matching:

\begin{definition}
We say a level-$\ell$ non-disjoint non-monotone non-proper matching $M_\ell$ at the beginning of an iteration of Algorithm~\ref{alg:ImprovedRecovery} is okay if there are at most $5k$ unmatched hashes or bad matches, i.e., matches between intervals in $F$ and $F'$ that are not-identical.
\end{definition}

We can now prove that each iteration of Algorithm~\ref{alg:ImprovedRecovery} works correctly with exponentially high probability in $k$, as long as it starts with an okay matching and as long as randomness in $R$ is independent between levels and sufficiently unbiased:

\begin{lemma}\label{lem:inductivecorrectnessofalgtwo}
Suppose the randomness in $R$ is at most $\exp(-ok)$-biased and independent between levels, where $o>1$ is a sufficiently large constant. If, at the beginning of iteration $\ell$ of Algorithm $2$, the matching $M_\ell$ is okay and $H_F[\ell,.]$ has been correctly recovered then, with prob $1-\exp(-ok)$, the matching $M_{\ell+1}$ is also okay. 
\end{lemma}
\begin{proof}
We want to bound the number of unmatched and bad matches in the matching $M_{\ell+1}$ produced by iteration $\ell$ of Algorithm $2$. By assumption $M_\ell$ is okay and thus has at most $5k$ unmatched or bad matches. 

The number of unmatched hashes in $M_{\ell+1}$ is easily bounded. In particular, for any unmatched hashes there exists a $k$-plausible disjoint monotone matching which leaves at most $k$ hashes unmatched, namely the one which matches all blocks in $F$ that do not have an edit in it. This is in particular true for the hashes that are unmatched after the matching $M_\ell$ has been cleaned up and transformed into a proper level-$\ell$ matching in Step $4$ of Algorithm $2$. Since $\Delta_\ell$ is the largest such matching it leaves at most $k$ hashes unmatched, which split into at most $2k$ unmatched hashes in $M_{\ell+1}$ in Step $6$ of Algorithm $2$. 

Next we bound the number of bad matches in $M_{\ell+1}$. There are two potential sources for bad matches, namely, they can either stem from bad matches in $M_\ell$ that are not identified as bad matches by $H_F[\ell,.]$, or they can be newly introduced by the matching $\Delta_\ell$. 

The expected number of bad matches of the first type is at most $5k \cdot 2^{-o} \cdot 2$ since each of the at most $5k$ bad matches in the okay matching $M_\ell$ has a non-matching hash in $H_F[\ell,.]$ with probability $2^{-o}$ and gets split into two potentially bad matches in $M_{\ell+1}$ if it goes undetected. The probability that this happens to more than $k/2$ matches leading to more than $k$ bad edges of this type in $M_{\ell+1}$ is at most $\exp(-ok)$ even if the probabilities in $R$ are $\exp(-ok)$-biased. 


Next we want to argue that, with probability $1 - \exp(-ok)$, the matching $\Delta_\ell$ introduces at most $k$ new bad matches which get doubled into at most $2k$ bad edges in $M_{\ell+1}$. For this we first bound the number of possible $\Delta_\ell$ matchings, given a fixed okay $M_\ell$ matching, by $\exp(k)$, and then take a union bound. To count the number of possible $\Delta_\ell$ matchings, given a fixed okay $M_\ell$, we specify such a matching by indicating which of the at most $5k$ unmatched hashes in $M_\ell$ are matched and what the offsets of their starting positions is. Since $|i_1 - i'_1| + |(|F| - i_{k'}) - (|F'| - i'_{k'})| + \sum_{j=1}^{k'-1} |(i_j - i'_j) - (i_{j+1} - i'_{j+1})|$, i.e., the sum of the differences between these offsets, is at most $k$ for every $k$-plausible matching the values and signs of these offsets can take on at most $\exp(k)$ different values. Overall there are therefore at most $\exp(k)$ different possibilities for $\Delta_\ell$ given $M_\ell$. Furthermore, the probability for any fixed such matching to be contain $k$ bad matches is at most $\exp(-ok)$, if the randomness in $R$ is independent from $M_\ell$ and at most $\exp(-ok)$-biased. A union bound over all $\exp(k)$ possibilities for $\Delta_\ell$ thus shows that with high probability at most $\exp(k - ok) = exp(-ok)$ there is no $\Delta_\ell$ matching which introduces more than $k$ new bad matches. 

Overall, with probability $1-\exp(-ok)$, this leads to at most $5k$ unmatched hashes or bad matches in $M_{\ell+1}$ at the end of iteration $\ell$ of Algorithm $2$, making it an okay matching as desired.
\end{proof}
\medskip

\Cref{lem:inductivecorrectnessofalgtwo} shows that our improved matching procedure in Algorithm~\ref{alg:ImprovedRecovery} is robust enough to tolerate hashes of constant size $o=\Theta(1)$. In fact, it guarantees that given an okay matching for level $\ell$ and the correctly recovered hashes for level $\ell$ a finer grained matching for level $\ell-1$ is computed which is okay, i.e., which allows all but $5k$ hashes of level $\ell+1$ to be guessed correctly. Algorithm~\ref{alg:ImprovedRecovery} thus achieved the crucial feat of reducing the edit distance document exchange problem to its much simpler Hamming type equivalent in which two parties hold a long string differing by at most $5k$ Hamming errors and one party wants to help the other learn its string.

\paragraph{Remark --  Impossibility of Reconciling $5k$ (Worst-Case) Hamming Errors with $o(k \log \frac{n}{k})$ bits}

\noindent In Algorithm $1$ the reduction to the Hamming problem was all that was needed. There, too, the recovery algorithm found a guess $\tilde{H}[\ell+1,.]$ for $H[\ell+1,.]$ which differed by at most $O(k)$ hashes. Both of these strings of hashes were over an alphabet of $o = \log \frac{n}{k}$ bits and one could then simply use the error correcting code idea of \cite{irmak2005improved} to recover $H[\ell+1,.]$ from $\tilde{H}[\ell+1,.]$ using $\Theta(ko) = \Theta(k \log \frac{n}{k})$ bits of additional information which could be put into $S_F$. Concretely, we used the non-systematic part of a systematic linear $[4k 2^{\ell}+100k,4k 2^{\ell},11k]$ error correcting code over $\F_{2^o}$ to send the equivalent of $O(k)$ hashes and recover the position and correct value for the $5k$ hashes differing between the matching generated guess and the true hashes. 

Unfortunately however, for $o=\Theta(1)$, such error correcting codes cannot exist and in fact it is easy to verify\footnote{Thanks to Xin Li and his group for pointing out this error in the preliminary draft of this work.} that it is impossible to reconciliate two parties holding $n$ long strings over some alphabet $\Sigma$ differing in any $k$ positions without sending at least $\Theta(k \cdot (\log \frac{n}{k} + \log (|\Sigma|-1)))$ bits, because the position of the differences can already encode $\log \binom{n}{k} = \Theta(k \log \frac{n}{k})$ bits. For the encoding used in $S_F$ this implies that either $\Theta(k \log \frac{n}{k})$ bits need to be put into $S_F$ per level to allow the recovery of the $\Theta(k)$ bad or missing hashes, as we do in Algorithm~\ref{alg:SimpleRecovery}, or one needs to have a better understanding of the distribution of the typical positions of bad hashes, together with a better coding scheme which exploits the lower entropy in this distribution to communicate efficiently. In particular, we would like to only use $\Theta(1)$ bits per bad hash to describe its position and correct value. This is what we do next.

\paragraph{Understanding the Distribution of Positions of Incorrectly Matched Substrings and Defining the Forest of Still Consistent Matches}

\noindent Suppose we run Algorithm~\ref{alg:ImprovedRecovery} for $\ell$ iterations. As proved in \Cref{lem:inductivecorrectnessofalgtwo}, with high probability, in each level the matching $\Delta_\ell$ adds at most $5k$ newly matched substrings. These substrings get split in two in every level thereafter or eliminated if non-matching hashes reveal that a match is inconsistent (proving that its guess was wrong). Each matched substring in level $\ell$ can thus be classified by the level $\ell' \leq \ell$ its first ancestor was generated, which of the at most $5k$ newly matched substrings in level $\ell'$ this ancestor was, and which of the at most $2^{\ell-\ell'}$ substrings stemming from this ancestor it is. In this way the set of all substrings matched throughout Algorithm~\ref{alg:ImprovedRecovery} can be naturally organized into $\leq 5k$ binary trees of depth $\ell-\ell'$ for each level $\ell' \leq \ell$. We call this the \emph{forest of all matches}.

Throughout Algorithm~\ref{alg:ImprovedRecovery} some of these matches are discovered to be inconsistent and removed from $M_{\ell}$. In particular, once Algorithm~\ref{alg:ImprovedRecovery} recovers the correct hashes $H[\ell,.]$ at the end of iteration $\ell-1$ it will, at the beginning of the next iteration, check for every matching edge in $M_{\ell}$ whether it is consistent and otherwise remove it from $M_{\ell}$ to make the matching proper. It is important to note that if a match is discovered to be bad then all ancestor matches in the forest of all matches are proven to be bad as well, despite their consistent hashes having failed to detect this badness at the time because of a hash collision. We say that a match is \emph{still consistent} if it has not (yet) been proven to be bad through an inconsistent hash of a descendant. The set of all substrings matched throughout Algorithm~\ref{alg:ImprovedRecovery} which are still consistent is similarly nicely organized as a forest of binary trees, where there leafs are exactly the matches/substrings in $M_{\ell}$. The main difference to the forest of all matches is that when a match in $M_{\ell}$ is discovered to be inconsistent then it and all its ancestor matches in its leaf-to-root path are removed. This cuts a tree of depth $d$ into up to $d-1$ trees, at most one for each depth smaller than $d$. The number of trees in level $\ell'$ therefore never exceeds the number of matches made in iteration $\ell'$ plus the number of bad matches (from previous iterations) in iteration $\ell'$. According to \Cref{lem:inductivecorrectnessofalgtwo} this is at most $5k$, with high probability. Throughout the rest of this paper we focus on the forest of still consistent matches. We say that a match in $M_{\ell}$ \emph{stems} from level $\ell' \leq \ell$ if its root in the forest of still consistent matches is a level $\ell'$ match. 

Since hashes fail independently with probability $\exp(-o)$ having a substring stemming from level $\ell'$ be discovered to be wrong has probability $\exp(-(\ell-\ell')o)$. Of course there are also more of these substrings, namely up to $5k 2^{\ell-\ell'}$ many of them. However, given that $o$ is a sufficiently large constant a union bound shows that one still expects most bad matches to be among the substrings stemming from higher levels with the expected number of bad hashes decaying exponentially with the level they are stemming from. This is quite intuitive, given that matches from these more recent iterations have not been included/tested by hashes quite as often. 

We will rely on this exponential concentration of bad hashes towards the much smaller number of positions corresponding to recent matches when trying to recover the correct $H[\ell+1,.]$ from 
the guesses for these level $\ell+1$ hashes generated by $M_{\ell+1}$. In particular, we identify a sufficiently small number of plausible guesses for sets of matches or positions in $H[\ell+1,.]$ which, with high probability, include at least one guess which \emph{covers} all inconsistent hashes. In fact, we will show that with high probability there is a guess which is a super-set of all bad hashes. For any such guessed set of positions for bad/inconsistent hashes we then enumerate all possible values for these positions to get a guess for the correct $H[\ell+1,.]$. We then use some extra hash (or hashes) of $H[\ell+1,.]$, which are included in $S_F$, to verify if which of the enumerated choices for $H[\ell+1,.]$ is correct.

\paragraph{Enumerating Plausible Guesses Using $t$-Witnesses}

To formally implement this intuition and strategy and to prove that it works we use combinatorial structures we call $t$-witnesses. They are a specially formated way of specifying some $\Theta(t)$ guesses for incorrect positions in $H[\ell,.]$ (or equivalently inconsistent matches). To specify a $t$-witness at level $\ell$, i.e., a guess of at most $\Theta(t)$ matches stemming from levels $\ell$ and above, we first specify a non-negative number for each of the last $\min\{l,t\}$ levels, i.e., for each integer $0 \leq i \leq \min\{\ell-1,t-1\}$ we specify an integer $0 \leq b_i \leq t$ with the restriction that $\sum_i i b_i \leq t$. As we will describe later these $b_i$ essentially specify the number of extra substrings stemming from level $\ell-i$ for which our guessed hash is not matching the actual hash for the next iteration because the substring is incorrect but has gone undetected so far. To specify which substrings among those in this level those are we have for each integer $0 \leq i \leq \min\{\ell-1,t-1\}$ a subset of positive integers $B_i \subseteq \{1,2,\ldots,t 2^i\}$ of size $|B_i| \leq \floor{t 2^{-i}} + b_i$.

\medskip

Next we explain how exactly a $t$-witness in level $\ell$ for $t > \Theta(k)$ encodes a set of at most $\sum_i |B_i| = \sum_i \floor{t 2^{-i}} + \sum_i b_i = 2t + t = 3t$ matches for a given matching $M_\ell$. Recall that these matches exactly correspond to leafs in the forest of still consistent matches. Process the $B_i$ sets from the largest $i$ to the smallest and process each $j \in B_i$ from the smallest to the largest. In particular, we start with the largest $i$ for which $B_i$ is non-empty and select the smallest integer $j \in B_i$. This specifies a leaf stemming from level $\ell-i$ in one of trees of depth $i$ in the forest of still consistent matches by simply taking the $\ceil{j/2^i}$th such tree and selecting its $((j \mod 2^i)+1)$th leaf (counting leafs in cut-out subtrees as well). Any leaf can be specified this way if there are at most $t$ trees. Before continuing to process the next (larger) $j$ (or smaller $i$ if there is no further integer in the current $B_i$) we cut all nodes from the chosen substring to its root, creating at most $i$ subtrees of smaller depth which we add to the corresponding levels. In essence we treat the match as if it was discovered to be inconsistent and update the forest of still consistent matches accordingly. We then continue similarly with the next guess. 

The reason for the cutting is that each incorrect substring from a level $\ell-i$ tree corresponds to $i$ failed hashes with the caveat that for two such strings these hashes might overlap. Cutting and reclassifying the cut-off trees and leafs/substrings as above makes sure that any substring specified by a $j \in B_i$ corresponds to $i$ disjoint failed hashes. 

It remains to analyze the number of such $t$-witnesses and to show that checking all $t$-witnesses for $t = 6k$ suffices to indeed check all typical ways in which hashes fail. 

\begin{lemma}\label{lem:numberofwitnesses}
For any $t$ the number of $t$-witnesses is at most $\exp(t)$.
\end{lemma}
\begin{proof}
The condition $\sum_{0 \leq i} i b_i = t$ implies that the sum of all $b_i$ for a $t$-witness is at most $t$. The number of possibilities  of different choices for setting the $b$-values, i.e., for distributing these these $t$ ``tokens'' over at most $t$ levels, is at most $\binom{2t}{t} = \exp(t)$. Furthermore, the number of possibilities to pick $\floor{t 2^{-i}} + b_i$ integers smaller than $t 2^i$ is at most $\binom{t 2^i}{\floor{t 2^{-i}} + b_i} = \Theta(4^i)^{\Theta(t) 2^{-i} + b_i} \leq \exp(t i 2^{-i} + i b_i)$. The total number of different $t$-witnesses for a given setting of $b$ values is thus at most $\prod_i \exp(t i 2^{-i} + i b_i) = \exp(t \sum_j j 2^{-j}) \exp(\sum_i i b_i) = \exp(t)$.
\end{proof}

\medskip

\begin{lemma}\label{lem:witnessprobability}
For $t = 6k$, with probability $1 - \exp(-ot)$, the set of substrings in $M_{\ell}$ that are bad can be covered by a $t$-witness.
\end{lemma}
\begin{proof}
Suppose the set $S$ of bad level $\ell$ matches in $M_{\ell}$ cannot be covered by a $t$-witness. This means that there exists as subset of levels $I \subseteq [\ell]$ such that for each $i \in I$ the number of bad matches in $M_{\ell}$ stemming from level $\ell - i$ (after cutting) is by $b_i > 0$ larger than $\floor{t 2^{-i}}$ where $\sum_{i\in I} i b_i \geq t$.

Following the argument from \Cref{lem:inductivecorrectnessofalgtwo} the number of incorrectly matched substrings in the last level is at most $5k$ with probability $1 - \exp(-ok) = 1 - \exp(-ot)$. In this case on can choose $b_0 = 0$ for any potential $t$-witness such that $I$ does not contain $0$.

Next we show that the probability for a given set of such bad matches to have survived up to iteration $\ell$ is at most $\exp(-ot)$.
In order for a specific substring that stems from level $\ell-i$ to fail the $i$ hashes including it in every level but the last one must have failed. The probability for this is $\exp(-io)$ for independent hashes, which is guaranteed through the cutting of overlapping hashes. Given that there are at most $t 2^i$ substrings that get processed in level $\ell-i$ the expected number of such substrings to be incorrect and not previously discovered in this level is $t\exp(-io)$ and the probability for $\floor{t 2^{-i}} + b_i$ such substrings to exist is at most $\exp(- i o b_i)$. The probability for a fixed set of matches as specified above to be a description of bad matches is thus at most $\prod_{i\in I} \exp(- i o b_i) = \exp(- o \sum_{i\in I} i b_i) = \exp(-ot)$.

According to \Cref{lem:numberofwitnesses} the number of $t$-witnesses is at most $\exp(t)$ and specifying a full witness gives rise to at least as many possibilities as just specifying the values for all $B_i$ with $i \in I$. Therefore A simple union bound over all such possibilities completes the proof that having a set of bad matches in $M_{\ell}$ which cannot be covered by a $t$-witness is at most $\exp(t) \cdot \exp(-ot) = \exp(-ot)$.  
\end{proof}

\medskip

Therefore, checking all $\exp(k)$ many $\Theta(k)$-witnesses suffices, with probability $1-\exp(-ok)$, to find (a superset of) the positions of the substrings which have a non-matching hash in this level. Even trying all $\exp(o)^{5k}$ possibilities for what the correct hash values are for each of these $\exp(k)$ guesses and verifying if it leads to a matching hash $h_i$ for the whole string of hashes works correctly because the probability that an incorrect guess has a matching hash $h_i$ is $\exp(-o'k)$, which even after a union bound over all $\exp(k)\exp(ok)$ guesses is negligible. 

Overall for any $k \in \Omega(\log \log n) \cap O(\log n)$ this makes Algorithm~\ref{alg:ImprovedRecovery} correct and efficient since each of the $O(\log n)$ levels succeeds with probability $1 - \exp(-ok) > 1 - \log^{O(1)} n$ and the number of guesses in each level one needs to try is at most $\exp(\log n) = n^{O(1)}$. In the case of $k < o(\log \log n)$ there are too many levels to simply do a union bound with the $\exp(-ok)$ failure probabilities for correctness and for larger $k$ the number of guesses needed makes the algorithm inefficient. These two problems are handled relatively easily as we show next. We first prove \Cref{thm:mainrandomized} for the case of small $k$, i.e., for $k = O(\log n)$:

\begin{proof}[Proof of \Cref{thm:mainrandomized} for $k = O(\log n)$.\ ]
We first note that for $o=O(1)$ the hashes in the summary $S_F$ are indeed of size $O(ok \log \frac{n}{k})$ as desired. Furthermore, given the constructions for $\eps$-biased distributions from \cite{naor1993small} one can specify the randomness for each level using $O(ok)$ bits or $O(ok \log \frac{n}{k})$ across all levels. We thus overall have a summary size of $O_o(k \log \frac{n}{k})$ as claimed. 

It furthermore follows almost immediately from \Cref{lem:inductivecorrectnessofalgtwo} that Algorithm~\ref{alg:ImprovedRecovery} is a successful decoding algorithm with probability $\log \frac{n}{k} \cdot \exp(-o k)$ as long as $R$ is chosen independently from an $\exp(-ok)$ biased distribution for each level. 

In particular, by induction on $\ell$, each iteration starts with a correct $H_F[\ell,.]$ and an okay matching $M_\ell$. This is true for $\ell=0$ because $H_F[0,.]$ is part of the summary of $F$ used as an input and $M_0$ is the empty $0$-level matching which consists of $4k$ unmatched hashes in $F$. For subsequent levels, we get from \Cref{lem:inductivecorrectnessofalgtwo} that $M_{\ell+1}$ is also okay. This then leads to a guess $\tilde{H}_F(\ell+1,.)$ for the level $\ell+1$ hashes which is correct up to $5k$ hashes. With probability $1 - \exp(-ok)$ these can be described by a $\Theta(k)$-witness by \Cref{lem:witnessprobability}. Guessing the correct hash values for this witness then leads to the correct hashes which is recognized by a matching of the hash $h'_i$. The probability that among the other $\exp(k)$ many $\Theta(k)$-witnesses each with $\exp(ok)$ guesses for their hash values there is an incorrect one which still matches $h'_i$ is $\exp(k) \exp(ok) \exp(-o'k) = \exp(-o'k) < \exp(-ok)$. In each level we thus recover the correct hashes $H_F(\ell+1,.)$ with probability $1-\exp(-ok)$. A union bound over all $\log \frac{n}{k}$ levels then leads to a failure probability of at most $\log \frac{n}{k} \cdot \exp(-o k)$. 

While this failure probability is $o(1)$ if $k = \omega(\log \log n)$ it is not quite as strong as the $\exp(k + \log n)$ failure probability claimed by \Cref{thm:mainrandomized} and furthermore becomes meaningless for even smaller $k$. Therefore, for $k < \frac{\log n}{\sqrt{o}}$ we modify Algorithm~\ref{alg:ImprovedRecovery} as follows: Enumerate over any subset $U$ of levels of size $|U| = O(\frac{\log n}{\sqrt{o}k})$ levels and instead of running Algorithm~\ref{alg:ImprovedRecovery} on this level try all $\exp(k)$ different $k$-plausible matchings as a possibility of $\Delta_\ell$. Given that there are at most $\binom{\log n}{|U|} = \exp(\log n \frac{\log ok}{\sqrt{o}k}) = n^{\frac{\log ok}{\sqrt{o}k}}$ many subsets of levels and at most $\exp(k|U|)=\exp(\frac{\log n}{\sqrt{o}})$ many matchings to try for all of these levels this requires at most $n^{\frac{\log o}{\sqrt{o}}}$ many different modified runs of Algorithm~\ref{alg:ImprovedRecovery} which is an essentially negligible overhead in the recovery time. Furthermore, the probability that none of these runs successfully recovers $F$ is at most the probability that of having $|U|$ failures in $\log \frac{n}{k}$ trials in which each trial succeeds independently with probability $\exp(-ok)$, and thus at most $\binom{\log n}{|U|} \exp(-ok)^|U| < 
n^{\frac{\log ok}{\sqrt{o}k}} \exp(-\log n \sqrt{o}) = n^{-\sqrt{o}}$. To see this we apply \Cref{lem:inductivecorrectnessofalgtwo} as before until the first iteration fails, which happens with independent probability of $\exp(-ok)$ for each iteration. We then look at the run in which the first failed iteration is the first iteration which the matching process is simply replaced by a ``guess'' for $\Delta_\ell$ and in particular when this guess is the correct matching. The following iterations in this run again fail independently with probability $\exp(-ok)$. If we continue to replace all failing iterations in this run we end up with a correct run, unless more than $|U|$ iterations fail independently. 

Over all these runs we get at most some $n^{\frac{\log o}{\sqrt{o}}}$ potential guesses for $F$. For each such guess we can check whether indeed a file $F$ was recovered with $ED(F',F) < k$. Adding an independent random hash of $F$ of size $\Theta(\log n)$ to the summary $S_F$ and checking whether it matches with what was recovered is sufficient to ensure that the algorithm, with high probability, only terminates and outputs a recovery once the correct $F$ is found. 
\end{proof}

\medskip

Next we proof \Cref{thm:mainrandomized} for the case of large $k = \omega(\log n)$. While here correctness and failure probabilities are not an issue, efficiency is. In particular trying all $\exp(k)$ guesses for $\Theta(t)$-witnesses becomes super-polynomial and thus intractable. The approach to avoid this is based on simple sampling. Instead of trying to guess all $\Theta(k)$ substrings with non-matching hashes and then verifying them via the hash $h'_i$ we instead randomly put all substrings in $k/\log n$ subsets. With high probability each subset contains at most $O(\log n)$ substrings that need correcting and the same arguments as before show that with high probability these can be specified with a $\Theta(\log n)$-witness. Trying the polynomially many such witnesses in each class and verifying them independently via separate hashes leads to a simple polynomial time computation which succeeds with high probability.

\begin{proof}[Proof of \Cref{thm:mainrandomized} for $k = \omega(\log n)$.\ ]
In addition to the $H$-hashes in the summary $S_F$ which have size $O(ok \log \frac{n}{k})$ as before some $\Theta(k) = \Omega(\log n)$ independent bits of randomness are added for each level which, using the $\eps$-biased distributions from \cite{naor1993small} are used to color each substring in each level with a uniformly random color between $0$ and $\ceil{\frac{k}{\log n}}$. For each level in $S_F$ there is furthermore a hash of size $O(o' \log n)$ bits added for each color class, hashing the string of hashes of the same color. Here $o' >> o$ is a sufficiently large constant. The summary size remains $O(k \log \frac{n}{k})$ bits as before.

The algorithm for recovery essentially also follows Algorithm~\ref{alg:ImprovedRecovery} except that in order to recover the correct new set of level $\ell+1$ hashes $H_F(\ell+1,.)$ in iteration $\ell$ we consider each color class separately. In particular the algorithm enumerates all $\Theta(\log n)$-witnesses with substrings in a single color class $c$ and all guesses for the correct hash values for them. \Cref{lem:numberofwitnesses} guarantees that there are only a polynomial number of such witnesses and thus a polynomial number of different guesses for the part of $H_F(\ell+1,.)$ colored with $c$. Which of these guesses is correct can be verified by the extra $O(\log n)$ size hash added to $S_F$ and this verification is correct with high probability. While it is clear that among the $5k$ bad substrings guaranteed by \Cref{lem:inductivecorrectnessofalgtwo} there will be at most $O(\log n)$ in each color class with high probability given that any such string ends up in a specific color class $c$ with probability $\frac{\log n}{k}$ it remains to be shown that these substrings are, with high probability, described by a $\Theta(\log n)$-witness. This follows in the same way as the proof of \Cref{lem:witnessprobability}: In order for a specific substring in level $\ell-i$ to be incorrect the $i$ hashes including it must have failed. The probability for this is $\exp(-io)$. The probability to get colored $c$ for a fixed $c$ is furthermore independent and $\Theta(\frac{\log n}{k})$ giving an overall probability of $\frac{\log n}{k} \exp(-io)$. Among the $O(k 2^i)$ substrings in level $\ell-i$ the expected number of substrings discovered to be incorrect in this level and be in color class $c$ is thus $\log n \exp(-io)$ and the probability for $|B_i| \leq \floor{\log n 2^{-i}} + b_i$ such substrings to exist is at most $\exp(- i o b_i)$. The probability for a fixed $O(\log n)$-witness to describe substrings with non-matching hashes is thus $\prod_i \exp(- i o b_i) = \exp(- o \sum_i i b_i) = \exp(-o \log n) = n^{-O(o)}$. For a sufficiently large $t = \Theta(\log n)$ no such $t$-witness which describes substrings with incorrect hashes exists which means that all the substrings whose hashes need correcting can be described by a $t'$-witness with $t' < t$ and are thus found and corrected by the recovery algorithm. Overall each iteration succeeds with high probability also giving the desired with high probability guarantee for the entire algorithm.
\end{proof}

\section{Derandomizing Algorithm 1 - The Deterministic Document Exchange Protocol of \Cref{thm:maindet} and New Error Correcting Codes}\label{sec:derandomization}

In this section we show how to derandomize Algorithm 1 and complete the proofs of \Cref{thm:maindet} and \Cref{thm:maincodes}. 

\subsection{Derandomizing Algorithm 1 and the Proof of \Cref{thm:maindet}}

As a first important step we show that one can significantly weaken the requirements Lemma~\ref{lem:okbiasednomatching} puts on the quality of the randomness provided. In particular, the conclusion of Lemma~\ref{lem:okbiasednomatching} holds with high probability even if we use a $(n^{-c})$-biased distribution. \\

\begin{lemma}\label{lem:polynbiasednomatching}
For every sufficiently large $c \geq 1$ it holds that if $o = c \log \frac{n}{k}$ and $R$ is sampled from an $n^{-2c}$-biased distribution of bits then for every level $i$ the probability that there exists a $k$-bad matching of $F$ under $\hash$ and $R$ is at most $n^{-\Omega(c)}$.
\end{lemma}

In order to prove this we apply a trick very similar to the long-distance property in \cite{haeupler2017synchronization3} used to derandomized synchronization strings. In particular, we show that we can restrict ourselves to use much smaller (sub-)matchings as witnesses:\\

\begin{lemma}\label{lem:largetosmallmatching}
If $F$ has a $k$-bad self-matching under $\hashRo$ in level $\ell$ than for any $1 \leq k' \leq \frac{k}{2}$ it also has a $k'$-bad self-matching under $\hashRo$ in level $\ell$ with $(i_{k'} - i_1) + (i'_{k'} - i'_1) \leq \frac{4k'}{k}n$.
\end{lemma}
\begin{proof}
Let $i_1,\ldots,i_{k},i'_1,\ldots,i'_{k} \in [1,n]$ be the indices for the $k$-bad matching of $F$ under $\hashRo$ in level $\ell$. Decompose this $k$-bad matching into $\floor{\frac{k}{k'}} \geq \frac{k}{2k'}$ many $k'$-bad matchings where $i_{(j-1)k'+1},\dots,i_{jk'}$ and $i'_{(j-1)k'+1},\dots,i'_{jk'}$ are the indices for the $j^{th}$ such matching. Denote with $\ell_j = (i_{jk'} - i_{j-1)k'+1}) + (i'_{jk'} - i'_{(j-1)k'+1})$ the total length of the $j^{th}$ matching. Note that the total length of all these $k'$-bad matchings sums up to at most $2n$. Therefore the shortest such matching is a $k'$-bad matching in $F$ with the desired length bound of at most $\frac{2n}{\frac{k}{2k'}}$. 
\end{proof}

\medskip

We can now prove Lemma~\ref{lem:polynbiasednomatching} in a similar way as Lemma~\ref{lem:okbiasednomatching}.

\begin{proof}[Proof of Lemma~\ref{lem:polynbiasednomatching}]
According to Lemma~\ref{lem:largetosmallmatching} it suffices to prove that with high probability $F$ does not have a $k'$-bad matching with the ``length restriction'' $(i_{k'} - i_1) + (i'_{k'} - i'_1) \leq \frac{4k'}{k}n$ for some $k'$. We choose $k' = \frac{\log n}{\log \frac{n}{k}}$ such that $ok' = c \log n$.
We again start with the case that $R$ is sampled from iid uniformly random bits. The probability for any two sub-strings of $F$ to have a bad hash is still $2^{-o}$. Furthermore, for a fixed $k'$-bad matching the probability that all matching pairs are bad under $\hash$ is $2^{-ok'} = n^{-c}$. There furthermore exist at most $n^2$ ways to pick $i_{k'}$ and $i'_{k'}$ and at most $\binom{\frac{4k'}{k}n}{2k'-2} = \left(\frac{n}{k}\right)^{O(k')}$ ways to pick the other $2k'-2$ indices given the length restriction for a total of $n^2 \cdot 2^{k' \log \frac{n}{k}} = n^{O(1)}$ potentially $k'$-bad matchings satisfying the length restriction. Choosing the constant $c$ large enough and taking a union bound over all these potentially bad matchings guarantees that, with high probability, no such matching is $k'$ bad which due to Lemma~\ref{lem:largetosmallmatching} guarantees that $F$ does not have a $k$-bad matching. Given that our argument for one $k'$-bad matching (again, similar to Lemma~\ref{lem:okbiasednomatching}) under a $R$ sampled from iid uniformly random bits only depends on $O(\log n)$ linear test on bits from $R$ sampling $R$ instead from an $(n^{-c})$-biased distribution does not change the above argument. 
\end{proof}

\medskip

The improvement of Lemma~\ref{lem:polynbiasednomatching} to polynomially biased random bits is particularly useful because there are simple constructions of such spaces with a polynomial size support which can be efficiently explored~\cite{naor1993small}. Furthermore, our definition of a $k$-bad matching in $F$ has the advantage that it only depends on $R$ and $F$ (and not $F'$). This allows one to determine independently of $F'$ whether a certain $R$ is a good choice for the ``randomness'' of Algorithm~$1$ when run on $F$. Putting all this together we get an efficient deterministic hashing scheme with summary size of $O(k \log^2 \frac{n}{k})$:\\

\begin{proof}[Proof of Theorem~\ref{thm:maindet}]
Given $F$ and $k$ we set $n = |F|$ and $o = c \log \frac{n}{k}$ for some small but sufficiently large constant $c$. Take a construction of an $|R|=no\log \frac{n}{k}$ long $(n^{-c})$-biased bit vector with polynomial support~\cite{naor1993small}. One by one (or in parallel) set $R$ to be one of these bit vectors and test whether under $R$ there exists a $k$-bad matching in $F$ under $R$ using a standard dynamic program. Do this until one setting of $R$ is found for which no such matching exists. The existence of such an $R$ is guaranteed by Lemma~\ref{lem:polynbiasednomatching}. The summary creation algorithm then uses this $R$ and $\hash$ and $o$ to create $S_F$. It also adds the $O(\log n)$ bit description of $R$ to $S_F$. 
The recovery algorithm is now simply Algorithm~$1$. Furthermore, because there is, by construction, no $k$-bad matching Lemma~\ref{lem:algonecorrectnesswithoutkmatch} guarantees that Algorithm~$1$ indeed terminates correctly. 
\end{proof}

\medskip

\subsection{New Error Correcting Codes for Insertions and Deletions and the Proof of \Cref{thm:maincodes}}

Finally one can use the deterministic document exchange protocol from \Cref{thm:maindet} and transform it into an error correcting code for insdel errors. For sake of complete we give here a complete proof of this (folklore) transformation:

\begin{proof}[Proof of Theorem~\ref{thm:maincodes}]
To encode $X$ we run the deterministic document exchange scheme from \Cref{thm:maindet} for edit distance $2k$ to obtain the summary $S_X$ consisting of $\Theta(k \log^2 \frac{n}{k})$ bits which can be converted into $O(k \frac{\log^2 \frac{n}{k}}{\log q} + k)$ symbols from $\Sigma$. If this is less than $k$ symbols we pad it to be $\Theta(k)$ symbols long. Next we encode these $O(k \frac{\log^2 \frac{n}{k}}{\log q} + k)$ symbols with any efficient error correcting block code $E$ which protects against a constant fraction of insdel errors, which is at least $2k$ insdels. For this we can use \cite{schulman1999asymptotically,GL-isit16,HaeuplerSTOC17p46} or \cite{haeupler2017synchronization2}. This increases the size by at most a constant. Overall we use the $r = O(k \frac{\log^2 \frac{n}{k}}{\log q} + k)$ symbols of $E(S_X)$ as the non-systematic part of the encoding $C(X)$.

Recovery now is also simple. Given a corrupted codeword $C'$ we interpret the first $n$ symbols as a corrupted version $X'$ of $X$ and the last $r$ symbols as a corrupted version $E'$ of $E(S_X)$. It is clear if $ED(C',C(X)) \leq k$ then both $E'$ and $X'$ have edit distance at most $2k$ from $E(S_X)$ and $X$ respectively. This allows us to decode $E'$ to $S_X$ and then use the document exchange recovery algorithm to recover $X$ from $X'$ and $S_X$.
\end{proof}

\medskip

The insdel codes from \Cref{thm:maincodes} improve over the error correcting code by Brakensiek, Guruswami and Zbarsky~\cite{brakensiek2016efficient} with redundancy $\Theta(k^2 \log k \log n)$ which are efficient under the strong assumption that $k$ is a fixed constant independent of $n$ and codes from Belazzougui's derandomized document exchange scheme~\cite{belazzougui2015efficient} which have a  redundancy of $\Theta(k^2 + k \log n)$. For the case of $k=\eps n$ the near optimal redundancy of $\Theta(\eps \log^2 \frac{1}{\eps} \cdot n)$ is a quadratic improvement over the codes of Guruswami et al.~\cite{guruswami2015deletion,GL-isit16} and Haeupler, Shahrasbi and Vitercik~\cite{haeupler2017synchronization2,HaeuplerSTOC17p46} which have redundancy $\Theta(\sqrt{\eps}\left(\log \frac{1}{\eps}\right)^{O(1)} \cdot n)$. The work of Cheng et al.~\cite{docexchangefocs}, which obtained \Cref{thm:maindet} independently and simultaneously, developed the ideas behind the deterministic document exchange even further and obtain non-systematic insdel codes with redundancy $O(k \log n)$. This is asymptotically optimal for any $k < n^{1 - \eps}$. Due to their non-systematic nature these codes do not correspond to a deterministic document exchange protocol. It remains an interesting open question whether \Cref{thm:maindet} can be improved and whether an efficient deterministic document exchange with optimal summary size, matching \Cref{thm:mainrandomized}, is possible.

\section*{Acknowledgments}

The author thanks Alon Orlitsky and Venkat Guruswami for introducing him to this problem. The author also thanks the group from \cite{docexchangefocs} for pointing out an error in the preliminary draft of this paper.





\end{document}